\newcommand{\longversion}[1]{#1}
\newcommand{\shortversion}[1]{}

\documentclass[12pt]{article}

\usepackage{amssymb,amsmath,amsthm}
\usepackage[usenames,dvipsnames]{color}
\usepackage{makecell}
\usepackage{multirow}
\usepackage{graphicx}

\usepackage{cancel}

\usepackage[usenames,svgnames,dvipsnames]{xcolor}

\usepackage{algorithm,algorithmicx,algpseudocode}

\algnewcommand{\Initialize}[1]{%
	\State \textbf{Initialize:}
	\Statex \hspace*{\algorithmicindent}\parbox[t]{.8\linewidth}{\raggedright #1}
}
\algdef{SE}[FUNCTION]{Function}{EndFunction}%
[2]{\algorithmicfunction\ \textproc{#1}\ifthenelse{\equal{#2}{}}{}{(#2)}}%
{\algorithmicend\ \algorithmicfunction}%

\algnewcommand{\LineComment}[1]{\State \(\triangleright\) #1}

\shortversion{
	\setlength\floatsep{.4\baselineskip plus 1pt minus 1pt}
	\setlength\textfloatsep{.4\baselineskip plus 1pt minus 1pt}
	\setlength\intextsep{.4\baselineskip plus 1pt minus 1pt}
}

\longversion{
	\usepackage{fullpage}
	\usepackage{charter,eulervm}
}

\def \nobreakseq {\nobreak \hskip 0pt \hbox}

\usepackage{tikz}
\usetikzlibrary{arrows,positioning}

\newdimen\prevdp
\def\leftlabel#1{\noalign{\prevdp=\prevdepth
		\kern-\prevdp\nointerlineskip\vbox to0pt{\vss\hbox{\ensuremath{#1}}}\kern\prevdp}}

\usepackage{url}

\usepackage{enumitem}

\usepackage{xspace}

\newcommand{\Pshort}{\ensuremath{\mathsf{P}}\xspace}

\newcommand{\NPC}{\ensuremath{\mathsf{NP}}\text{-complete}\xspace}
\newcommand{\NPH}{\ensuremath{\mathsf{NP}}\text{-hard}\xspace}

\newcommand{\el}{\ensuremath{\ell}\xspace}
\newcommand{\suc}{\ensuremath{\succ}\xspace}

\newcommand{\BigO}{\ensuremath{\mathcal{O}}\xspace}
\newcommand{\true}{\textsc{true}\xspace}
\newcommand{\false}{\textsc{false}\xspace}
\newcommand{\pr}{\ensuremath{\prime}\xspace}
\newcommand{\prr}{{\ensuremath{\prime\prime}}\xspace}

\newcommand{\PE}{\textsc{Preference Elicitation}\xspace}
\newcommand{\SCTR}{\textsc{Weakly Single Crossing Tree Recognition}\xspace}
\newcommand{\SCSR}{\textsc{Weakly Single Crossing Tree Recognition}\xspace}

\let\oldlambda\lambda
\renewcommand{\lambda}{\ensuremath{\oldlambda}\xspace}
\let\oldalpha\alpha
\renewcommand{\alpha}{\ensuremath{\oldalpha}\xspace}
\let\oldDelta\Delta
\renewcommand{\Delta}{\ensuremath{\oldDelta}\xspace}

\newcommand{\YES}{{\sc yes}\xspace}
\newcommand{\NO}{{\sc no}\xspace}

\newcommand{\Query}{\textsc{Query}\xspace}

\newcommand{\maj}{\ensuremath{\mathtt{Maj}}}

\renewcommand{\AA}{\ensuremath{\mathcal A}\xspace}

\newcommand{\CC}{\ensuremath{\mathcal C}\xspace}
\newcommand{\DD}{\ensuremath{\mathcal D}\xspace}
\newcommand{\EE}{\ensuremath{\mathcal E}\xspace}

\newcommand{\GG}{\ensuremath{\mathcal G}\xspace}

\newcommand{\LL}{\ensuremath{\mathcal L}\xspace}
\newcommand{\MM}{\ensuremath{\mathcal M}\xspace}

\newcommand{\OO}{\ensuremath{\mathcal O}\xspace}
\newcommand{\PP}{\ensuremath{\mathcal P}\xspace}
\newcommand{\QQ}{\ensuremath{\mathcal Q}\xspace}
\newcommand{\RR}{\ensuremath{\mathcal R}\xspace}

\newcommand{\TT}{\ensuremath{\mathcal T}\xspace}

\newcommand{\VV}{\ensuremath{\mathcal V}\xspace}
\newcommand{\WW}{\ensuremath{\mathcal W}\xspace}
\newcommand{\XX}{\ensuremath{\mathcal X}\xspace}

\newcommand{\PPP}{\ensuremath{\mathfrak P}\xspace}

\renewcommand{\lll}{\ensuremath{\mathfrak l}\xspace}

\newcommand{\ppp}{\ensuremath{\mathfrak p}\xspace}

\usepackage{nicefrac}

\newcommand{\nfrac}{\nicefrac}

\newtheorem{proposition}{\bf Proposition}
\newtheorem{observation}{\bf Observation}
\newtheorem{theorem}{\bf Theorem}
\newtheorem{lemma}{\bf Lemma}
\newtheorem{corollary}{\bf Corollary}
\newtheorem{definition}{\bf Definition}

\newtheorem{problemdefinition}{\bf Problem Definition}
\newtheorem{property}{\bf Property}

\newcommand{\eps}{\ensuremath{\varepsilon}\xspace}
\renewcommand{\epsilon}{\eps}

\newcommand{\ignore}[1]{}

\renewcommand{\ge}{\geqslant}
\renewcommand{\le}{\leqslant}

\usepackage{cleveref}

\crefname{theorem}{Theorem}{Theorems}
\crefname{observation}{Observation}{Observations}
\crefname{lemma}{Lemma}{Lemmata}
\crefname{corollary}{Corollary}{Corollaries}
\crefname{proposition}{Proposition}{Propositions}
\crefname{definition}{Definition}{Definitions}
\crefname{claim}{Claim}{Claims}
\crefname{reductionrule}{Reduction rule}{Reduction rules}
\crefname{problemdefinition}{Problem Definition}{Problem Definition}
\crefname{property}{Property}{Properties}

\title{Recognizing and Eliciting Weakly Single Crossing Profiles on Trees}

\author{
    Palash Dey\\Indian Institute of Technology, Kharagpur\\palash.dey@cse.iitkgp.ac.in
}

\begin{document}

\maketitle

\begin{abstract}
	Single crossing profiles on trees are not downward closed --- a sub-profile of a single crossing profile on trees is not necessarily a single crossing profile on trees. We define weakly single-crossing profiles on trees to be all single crossing profiles on trees and their sub-profiles thereby restoring downward closedness. We design a polynomial-time algorithm for recognizing these profiles. We then develop an efficient elicitation algorithm for this domain which works even if the preferences can be accessed only sequentially and the underlying single-crossing tree structure is not known beforehand. We complement our algorithmic results by proving a matching lower bound on the query complexity of our elicitation algorithm when the number of voters is large compared to the number of candidates. We also prove a lower bound of $\Omega(m^2\log n)$ on the number of queries that any algorithm needs to ask to elicit single-crossing profile when random queries are allowed. This resolves an open question in~\cite{deycross} and proves optimality of their preference elicitation algorithm when random queries are allowed.
\end{abstract}

\section{Introduction}

Aggregating preferences of a set of agents is a common problem in multiagent systems. In a typical setting, we have a set of $m$ {\em candidates}, a set of $n$ {\em voters} each of whom possesses a {\em preference} which is a linear order (reflexive, anti-symmetric, and transitive relation) over the set of candidates, and we would like to aggregate these preferences into one preference which intuitively ``reflects'' the preference of all the voters. However, it is well-known that the pairwise majority relation of a set of preferences may often be intransitive due to existence of the {\em Condorcet cycles} --- a set of candidates $\{c_i: i\in\{1,\ldots,\el\}\}$ where $c_i$ is preferred over $c_{i + 1}$ by a majority of the voters for every $i\in\{1,\ldots,\el-1\}$ and $c_\el$ is preferred over $c_1$ (see \cite{moulin1991axioms}). Consecutively, a substantial amount of research effort has been devoted to finding interesting restrictions on the preferences of the voters which ensure transitivity of the majority relation (see \cite{merlin2004domain} for a survey).

Among the most widely used domains in social choice theory are the {\em single-peaked} and {\em single-crossing} domains. Introduced by Black~\cite{black1948rationale}, the single-peaked domain not only satisfies transitivity of the majority relation but also captures the essence of many election scenarios including political elections~\cite{hinich1997analytical}. Intuitively, a {\em profile}, which is a tuple of preferences of all the voters, is called single-peaked if the candidates can be arranged in a linear order (often called societal axis or harmonious order) and every voter prefers candidates ``closer'' to her most preferred candidate than candidates ``far'' from her most preferred candidate in the societal axis~[refer to~\cite{mas1995microeconomic} for a formal definition of single-peaked preference profile]. The notion of single-peakedness has subsequently been generalized further, often at the cost of the transitivity of the majority relation. For example, the popular single-peaked domain on trees~\cite{demange1982single} only guarantees existence of a {\em Condorcet winner} (for an odd number of voters) and the notion of single-peaked width~\cite{CornazGS12} does not even guarantee existence of a Condorcet winner. A Condorcet winner is a candidate who is preferred over every other candidate by a majority of the voters.

Mirrlees~\cite{mirrlees1971exploration} proposed the single-crossing domain where voters (instead of candidates as in the case of single-peaked domain) can be arranged in a linear order so that, for every two candidates $x$ and $y$, all the voters who prefer $x$ over $y$ appear consecutively. Other than guaranteeing transitivity property of the majority relation, the single-crossing domain has found wide applications in economics~\cite{diamond1974increases}, specially in redistributive income taxation~\cite{roberts1977voting,meltzer1981rational}, trade union bargaining~\cite{blair1984labor}, etc. Single-crossing property has further been generalized with respect to {\em median graphs} while maintaining the transitivity of the majority relation~\cite{demange2012majority}. A graph is called a median graph if, for every three nodes, there exists a unique node which is present in the shortest paths between all three pairs of nodes. Trees, hypercubes, etc. are important examples of median graphs. A profile \PP is called single crossing on trees if there exists a tree \TT on \PP such that every sub-profile along any path in \PP is single crossing with respect the induced path order. Single-crossing profiles also have many real-world applications. For example, Kung applied single-crossing property on trees to the study of networks~\cite{kung2015sorting}. Clearwater et al.~\cite{clearwater2015single} show a polynomial-time algorithm to recognize whether or not a given profile is intermediate with respect to some median graph.

\subsection{Motivation and Contribution}

A domain is called ``downward closed'' --- if every sub-profile (with respect to the voters, throughout this paper) of any profile in the domain, also belongs to the domain. Downward closedness is a desirable property of any domain since some voters often do not turn up and we have to work with a sub-profile of the profile of population. Popular domains like single-peaked, single-crossing, etc. satisfy this property. Notable exceptions are top monotonic profiles~\cite{barbera2011top}, single-crossing profiles on trees, etc. To see that the domain of single-crossing profiles on trees is not downward closed, we refer to \Cref{fig:ex}.
\begin{figure}
	\centering
	\begin{tikzpicture}[sibling distance=9em,
		every node/.style = {shape=rectangle, rounded corners,
			draw, align=center}]]
		\node {$v_1: a\suc b\suc c\suc d\suc e\suc f$}
		child { node {$v_2: b\suc a\suc c$\\$\suc d\suc e\suc f$} }
		child { node {$v_3: a\suc b\suc d$\\$\suc c\suc e\suc f$} }
		child { node {$v_4: a\suc b\suc c$\\$\suc d\suc f\suc e$}
		};

	\end{tikzpicture}
	\caption{Profile $(v_1, v_2, v_3, v_4)$ is single-crossing on trees. However, $(v_2,v_3,v_4)$ is not single-crossing with respect to any tree. We call profiles like $(v_2,v_3,v_4)$ where some preferences can be added to make it single-crossing on trees, weakly single-crossing profile on trees.}\label{fig:ex}
\end{figure}
We remove this drawback of single-crossing profiles on trees by defining {\em weakly single-crossing profiles on trees}. We call a profile weakly single-crossing on trees if it is a sub-profile of a single-crossing profile on trees. Interestingly, we show that weakly single-crossing profile on trees inherit important properties like transitivity of the majority relation for an odd number of voters [\Cref{prop:tran}] and thus many otherwise intractable voting rules, for example Kemeny, Dodgson, etc., can be computed in polynomial amount of time. We then show in \Cref{cor:recognizing} that weakly single-crossing profiles can also be recognized in polynomial amount of time.

Another important question about a domain is if the profiles can be elicited by asking a few queries. Indeed, in many applications of social choice theory, for example, meta-search engines~\cite{Dwork}, spam detection~\cite{Cohen}, and computational biology~\cite{JacksonSA08}, the number of candidates is huge and thus, simply asking the agents to reveal their preferences is impractical. \cite{Conitzer09} showed that, for the domain of single-peaked profiles, we can elicit the preferences of a set of agents by asking a small number of comparison queries. Recently, a similar study has been carried out for the single-peaked domain on trees~\cite{deypeak} and the single-crossing domain~\cite{deycross}. This motivates us to study preference elicitation for the weakly single-crossing domain on trees. We present a polynomial time algorithm with query complexity $\OO(mn + \min\{m^2, n\} m\log m)$ for eliciting weakly single-crossing profiles on trees, even if we do not a priori know any tree with respect to which a super-profile of the input profile is single-crossing on trees and we are only allowed to access the preferences in an arbitrary (unknown) sequential order [\Cref{thm:scsequnknownub,cor:elicitsubdomain}]. We note that the query complexity of our preference elicitation algorithm matches (up to constant factors) with the query complexity of preference elicitation algorithm for single-crossing profiles due to \cite{deycross}. We complement the query complexity of our preference elicitation algorithm by showing that any preference elicitation algorithm for single-crossing profiles on trees has query complexity $\Omega(\max\{m\log m, mn\})$, even if the input profile is single-crossing with respect to a known tree and we are allowed to query preferences randomly and interleave the queries to different preferences arbitrarily [\Cref{thm:scrandomknownlb}].

We finally improve the $\Omega(m\log n)$ query complexity lower bound in~\cite{deycross} for eliciting single-crossing profiles when random queries are permitted to $\Omega(m^2\log n)$ thereby showing optimality of their algorithm [\Cref{thm:lb_sc}]. In summary, our work shows that weakly single-crossing profiles on trees have all the desirable properties of single-crossing profiles and moreover, it is downward closed.

\section{Related Work}

Popular domains often enjoy many desirable properties. Existence of strategy-proof voting rules, polynomial time tractability of the problem of finding the winner for important voting rules like Kemeny~\cite{kemeny1959mathematics,levenglick1975fair,brandt2015bypassing,DBLP:journals/iandc/FaliszewskiHHR11}, Dodgson~\cite{dodgson1876method,black1958theory}, Chamberlin-Courant~\cite{betzler2013computation,skowron2013complexity,clearwater2015single}, etc. are prominent examples of such properties. Domain restrictions in the presence of incomplete votes (where votes may be partial orders instead of complete orders) have also received significant research attention in recent times. Lackner~\cite{Lackner14} showed that determining whether a set of incomplete votes is single-peaked is an \NPC problem. Elkind et al.~\cite{DBLP:conf/aaai/ElkindFLO15} carried out a similar study for the single-crossing domain and showed that although the problem of finding whether a set of incomplete votes is single-crossing is \NPC in general, it admits polynomial time algorithms for top orders and for various other practically appealing cases. In an approval voting scenario, every voter simply approves a subset of candidates instead of specifying a complete ranking of the candidates. Closeness of a profile to various domains also forms an active body of current research. The work of \cite{DBLP:conf/aaai/ErdelyiLP13} and \cite{DBLP:journals/mss/BredereckCW16} showed that the problem of finding the distance of a profile from the single-peaked and single-crossing domains is often \NPH and occasionally polynomial time solvable under various natural notions of distance. \cite{DBLP:conf/aaai/ElkindL14} provided approximation and fixed parameter tractable algorithms for the problems of finding the minimum number of votes/candidates that need to be deleted so that the resulting profile belongs to some specific domain --- their algorithms work for any domain which can be characterized by forbidden configurations and thus, in particular, for both the single-peaked and single-crossing domains. \cite{Faliszewski:2014:CMA:2566265.2566291} studied various manipulation, control, and bribery problems in nearly single-peaked domain and showed many interesting results; for example, they proved that some control problems suddenly become \NPH even in the presence of only one maverick (whose vote does not satisfy the single-peaked property) whereas many other manipulation problems continue to be polynomial time solvable with a reasonable number of mavericks. We refer the reader to the recent survey article by \cite{DBLP:journals/corr/abs-2205-09092} and references therein for a more complete picture of recent research activity in preference restrictions in computational social choice theory.

\section{Preliminaries}

For a positive integer $\el$, we denote the set $\{1, \ldots, \el\}$ by $[\el]$. For a set \XX and an integer $k$, we denote the set of all possible subsets of \XX of size $k$ by $\PPP_k(\XX)$.

Let $\VV = \{v_i: i\in[n]\}$ be a set of $n$ {\em voters} and $\CC = \{c_j: j\in[m]\}$ be a set of $m$ {\em candidates}. If not mentioned otherwise, we denote the set of candidates, the set of voters, the number of candidates, and the number of voters by \CC, \VV, $m$, and $n$ respectively. Every voter $v_i$ has a {\em preference} $\suc_i$ which is a linear order over the set \CC of candidates. We say voter $v_i$ prefers a candidate $x\in\CC$ over another candidate $y\in\CC$ if $x\suc_i y$. We denote the set of all preferences over \CC by $\LL(\CC)$. For a preference ${\suc}\in\LL(\CC)$ and an integer \el, we denote \el copies of $\suc$ by $\suc^\el$. The $n$-tuple $(\suc_i)_{i\in[n]} \in\LL(\CC)^n$ of the preferences of all the voters is called a {\em profile}. For a subset $M\subseteq[n]$, we call $(\suc_i)_{i\in M}$ a sub-profile of $(\suc_i)_{i\in[n]}$. We often view a profile \PP as a multi-set consisting of the preferences in \PP to avoid use of cumbersome notations. The view of a profile we are considering will be clear from the context. The single-crossing domain is defined as follows.

\begin{definition}[Single Crossing Profile]
	A profile $\PP = (\succ_i)_{i\in[n]}$ of $n$ preferences over a set \CC of candidates is called a {\em single-crossing profile} if there exists a permutation $\sigma$ of $[n]$ such that, for every two distinct candidates $x, y\in\CC$, whenever we have $x\succ_{\sigma(i)} y$ and $x\succ_{\sigma(j)} y$ for two integers $i$ and $j$ with $1\le \sigma(i)< \sigma(j)\le n$, we have $x\succ_{\sigma(k)} y$ for every $\sigma(i)\le \sigma(k)\le \sigma(j)$.
\end{definition}

Demange generalizes the single-crossing domain to the single-crossing domain on median graphs in \cite{demange2012majority}. A graph $\GG=(\VV,\EE)$ is called a {\em median graph} if, for every three nodes $u, v, w\in\VV$, there exists a unique vertex $m(u, v, w)$, called the median of the nodes $u, v,$ and $w$, which is present in a shortest path between $u$ and $v$, $v$ and $w$, and $w$ and $u$. Trees and hypercubes are important examples of median graphs. In this work, we will be concerned with trees only. A {\em tree} is a connected acyclic graph. A {\em star} is a tree where there is a {\em central node} with whom every other node is connected by an edge. We refer to \cite{diestelBook} for common terminologies of trees. Given a tree \TT and a node $u\in\TT$, we denote the tree rooted at $u$ by $\TT[u]$. The single-crossing domain on trees is defined as follows.

\begin{definition}[Single Crossing Profile on Trees]\label{def:singlecrossingtree}
	A profile \PP is called {\em single-crossing on trees} if there exists a tree \TT on the set of voters such that every sub-profile along every path of \TT is single-crossing with respect to that path.
\end{definition}

Let a profile \PP be single-crossing with respect to a tree \TT. We call \TT a {\em single-crossing tree} of \PP. We denote the preference associated with a voter (which is a node in \TT) $u\in\TT$ by $\suc_u$. We denote a voter in \TT associated with a preference ${\suc}\in\PP$ by $u_\suc$. An equivalent condition for a profile to be single-crossing with respect to a tree \TT is that, for every pair of candidates $x, y\in\CC$, there exists at most one edge in the cut $(\VV_{x\suc y}, \VV_{y\suc x})$, where $\VV_{x\suc y}$ is the set of voters who prefer $x$ over $y$~\cite{clearwater2015single}. We now {\em generalize} the single-crossing domain on trees to the weakly single-crossing domain on trees as follows.

\begin{definition}[Weakly Single Crossing Profile on Trees]
	The weakly single-crossing domain on trees is the set of all profiles \PP such that there exists a profile $\PP^\pr$ which contains \PP and is itself single-crossing with respect to a tree \TT.
\end{definition}

\subsection{Problem Formulation} We study the following problem for recognizing weakly single-crossing profiles on trees.

\begin{problemdefinition}{\bf(\SCTR)}
	Given a profile \PP, does \PP belong to the weakly single-crossing domain on trees?
\end{problemdefinition}

Suppose we have a profile \PP with $n$ voters and $m$ candidates. Let us define a function $\text{\Query}(x \succ_\el y)$ for a voter \el and two different candidates $x$ and $y$ to be \true if the voter \el prefers the candidate $x$ over the candidate $y$ and \false otherwise. We now define the preference elicitation problem.

\begin{problemdefinition}[\PE]
	Given an oracle access to $\text{\Query}(\cdot)$ for a profile \PP, find \PP.
\end{problemdefinition}

For two distinct candidates $x, y\in \CC$ and a voter \el, we say a \PE algorithm \AA ~{\em compares} candidates $x$ and $y$ for the voter \el, if \AA makes a call to either $\text{\Query}(x \succ_\el y)$ or $\text{\Query}(y \succ_\el x)$. We define the number of queries made by the algorithm \AA, called the {\em query complexity} of \AA, to be the number of distinct tuples $(\el, x, y)\in \VV\times\CC\times\CC$ with $x\ne y$ such that the algorithm \AA compares the candidates $x$ and $y$ for the voter \el.\longversion{ Notice that, even if the algorithm \AA makes multiple calls to \Query($\cdot$) with same tuple $(\el, x, y)$, we count it only once in the query complexity of \AA. This is without loss of generality since we can always implement a ``wrapper'' around the oracle which memorizes all the calls made to the oracle so far and whenever it receives a duplicate call, it replies from its memory without ``actually'' making a call to the oracle. This allows us to simplify the description of our algorithm.}

The following observation is immediate from standard sorting algorithms like merge sort~\cite{cormen2009introduction}.

\begin{observation}\label{obs:naive}
	There is a \PE algorithm for eliciting one preference with query complexity $\BigO(m\log m)$.
\end{observation}

\subsection{Model of Input for \PE} There are two prominent models for accessing the preferences in the literature (see \cite{deycross}). In the {\em random access model}, we are allowed to query any preference at any point of time. Moreover, we are also allowed to interleave the queries to different voters. In the {\em sequential access model}, voters are arriving in a sequential manner one after another to the system. Once a voter \el arrives, we can query her preference as many times as we like and then we ``release'' the voter \el from the system to grab the next voter in the queue. Once the voter \el is released, its preference can never be queried again.

\section{Recognition Algorithm}

In this section, we first present some structural properties of weakly single-crossing profiles on trees and then exploit them to design our polynomial time algorithm for \SCTR. At a high-level, our algorithm for recognizing weakly single-crossing profiles on trees first find a set of preferences which when included in the input preference profile, must be a single-crossing profile on trees if the input preference profile was indeed weakly single-crossing on trees. Towards that, we investigate some structural aspects of single-crossing and weakly single-crossing preference profiles on trees. These structural results will be crucial to decide the set of preferences that we will add to the input preference profile to make it single-crossing on trees (if at all it is possible, that is, if the input preference profile is weakly single-crossing on trees). We now present our algorithm in detail.

We begin with bounding the number of distinct preferences in a profile which is single-crossing with respect to some tree. We remark that the exact same bound of \Cref{lem:sizeupperbound} is known for the single-crossing profiles~\cite{DBLP:journals/scw/ElkindFS20}.

\begin{lemma}\label{lem:sizeupperbound}
	If a profile \PP of distinct preferences is single-crossing with respect to a tree \TT, then $|\PP|\le {m\choose 2} + 1.$
\end{lemma}

\begin{proof}
	Let $e=(u,v)\in\TT$ be any edge of $\TT=(\VV,\EE)$. Since the preferences in \PP are pairwise distinct, the set $\DD_e = \{\{x,y\}\in\PPP_2(\CC): {\suc_u} \text{ and } {\suc_v} \text{ order $x$ and $y$ differently}\}$ of pairs of candidates that are ordered differently by the voters $u$ and $v$ is nonempty. Since the profile \PP is single-crossing with respect to \TT, for any two distinct edges $e_1, e_2\in\TT$, we have $\DD_{e_1}\cap\DD_{e_2}=\emptyset$. We also have $\cup_{e\in\TT}\DD_e \subseteq \PPP_2(\CC)$. Now we bound $|\PP|$ as follows.
	\[ |\PP| = |\EE|+1 \le {m\choose 2} + 1 \]
\end{proof}

We note that the set of weakly single-crossing profiles on trees is a strict superset of single-crossing profiles on trees, which in turn a strict superset of single-crossing profiles. Hence, \Cref{lem:sizeupperbound} does not follow from the known $\left({m\choose 2} + 1\right)$ upper bound for single-crossing profiles~\cite{deycross}. The fact that the bound turns out to be same is a mere coincidence.

The following result which is immediate from the proof of Lemma 3.7 in \cite{clearwater2015single}, simplifies a lot of our proofs.

\begin{lemma}\label{prop:distinct}
	Let $\PP = (\suc_i^{\el_i})_{i\in[n]}, \el_i>0, \forall i\in[n],$ be a profile with $\suc_i \ne \suc_j, \forall i\ne j$ and $\PP^\pr = (\suc_i)_{i\in[n]}$ be the profile resulting from \PP after removing all the duplicate preferences. Then \PP is single-crossing with respect to some tree if and only if $\PP^\pr$ is single-crossing with respect to some tree. Therefore, \PP is a weakly single-crossing profile on trees if and only if $\PP^\pr$ is a weakly single-crossing profile on trees. Moreover, given the tree with respect to which the profile \PP is single-crossing, we can construct another tree with respect to which $\PP^\pr$ is single-crossing in polynomial amount of time and vice versa.
\end{lemma}

We next define the majority relation of a set of preferences which helps us formulate an important property of single-crossing domain with respect to trees. We call that property the {\em triad majority} property. This property will in turn help us in defining and finding what we call the {\em single-crossing tree closure} of a weakly single-crossing profile on trees. The notion of a single-crossing tree closure plays a central role in our recognition and elicitation algorithms.

\begin{definition}[Majority Relation]
	Given a profile \PP of $n$ preferences $(\suc_i)_{i\in[n]}\in\LL(\CC)^n,$ we call the relation $\suc\; = \maj(\PP) = \{x \suc y: x, y\in\CC, |\{i\in[n]: x\suc_i y\}| > \nfrac{n}{2}\}$ the majority relation of the profile $\PP.$ If the majority relation $\maj(\PP)$ of a profile \PP turns out to be a linear order, we say that a majority order of $\PP$ exists and we call $\maj(\PP)$ the majority order of $\PP.$
\end{definition}

\begin{property}[Triad Majority Property]\label{property}
	A profile \PP is said to satisfy the triad majority property if, for every three preferences (not necessarily distinct) ${\suc_1},{\suc_2},{\suc_3}\in\PP,$ the majority relation $\maj((\suc_i)_{i\in[3]})$ of ${\suc_1},{\suc_2},{\suc_3}\in\PP,$ is a linear order.
\end{property}

For example, the voting profile $(a\suc b\suc c, b\suc c\suc a)$ over the set $\{a,b,c\}$ of candidates satisfies the triad majority property. However, the voting profile $(a\suc b\suc c, b\suc c\suc a, c\suc a\suc b)$ does not satisfy the triad majority property. This is so because the majority relation for any three preferences from the profile $(a\suc b\suc c, b\suc c\suc a)$ is transitive but the majority relation of the three preferences in the profile  $(a\suc b\suc c, b\suc c\suc a, c\suc a\suc b)$ is $a\suc b, b\suc c,$ and $c\suc a$ which is not transitive. We prove next that if a preference profile satisfies the triad majority property, then its majority relation must be transitive.


\begin{proposition}\label{prop:tran}
	If any preference profile satisfies the triad majority property, then its majority relation is transitive.
\end{proposition}

\begin{proof}
	We will prove the contra-positive of this statement. Let \PP be any profile whose majority relation is not transitive. We will prove that \PP does not satisfy the triad majority property. We first observe that \PP must have at least $3$ preferences since the majority relation of any preference profile with at most $2$ preferences is transitive. We now consider two cases:

	{\bf Case I:} {\it Suppose \PP has an odd number of preferences.} We observe that since the number of preferences is an odd integer, no two candidates tie and thus the majority graph is a tournament. Then there exists three candidates, say $a,b,c$ such that a majority of voters prefer $a$ over $b$, $b$ over $c$, and $c$ over $a$ --- this is because a tournament graph has a directed cycle if and only if it has a directed triangle. Then there exists a preference $\succ_1\in P$ such that $a\succ_1 b\succ_1 c$. This is so since more than half of the preferences have $a$ over $b$ and more than half of the preferences have $b$ over $c$. Similarly, we have $\succ_2,\succ_3\in P$ such that $b\succ_2 c\succ_2 a$ and $c\succ_3 a\succ_3 b$. We observe that the majority relation of $\suc_1, \suc_2,$ and $\suc_3$ is not linear and thus \PP does not satisfy the triad majority property.

	{\bf Case II:} {\it Suppose \PP has an even number of preferences.} Since the majority relation of \PP is not transitive, there exists $k$ different candidates, say $c_1,\ldots,c_k$ such that $c_i$ is preferred over $c_{i+1}$ in a majority of preferences in \PP for every $i\in[k-1]$ and $c_k$ is preferred over $c_1$ in a majority of preferences in \PP. Let $\suc\in\PP$ be an arbitrary preference in \PP. Let us consider the preference profile $\QQ=\PP\setminus\{\suc\}$. Since \PP has an even number of preferences, \QQ has an odd number of preferences. Moreover, since \PP has an even number of preferences and \QQ has all the preferences in \PP except \suc, $c_i$ is preferred over $c_{i+1}$ in a majority of preferences in \QQ for every $i\in[k-1]$ and $c_k$ is preferred over $c_1$ in a majority of preferences in \QQ. Now from case I, we conclude that \QQ and thus \PP do not satisfy the triad majority property.
\end{proof}

We now show that every single-crossing profile on trees, satisfies the triad majority property.

\begin{lemma}\label{lem:necessary}
	Let \PP be a single-crossing profile with respect to a tree \TT. Then \PP satisfies the triad majority property and $\maj((\suc_i)_{i\in[3]})\in\PP$ for every $\suc_1,\suc_2,\suc_3\in\PP$.
\end{lemma}

\begin{proof}
	For any three preferences $\suc_i, i\in[3],$ let the nodes associated with $\suc_i\in\PP, i\in[3],$ in \TT be $u_{\suc_i}\in\TT, i\in[3]$. Let $u$ be the unique node in \TT that lies in the shortest path between $u_{\suc_1}$ and $u_{\suc_2}, u_{\suc_2}$ and $u_{\suc_3},$ and $u_{\suc_3}$ and $u_{\suc_1}$ in \TT. We claim that $\maj(\suc_1, \suc_2, \suc_3) = \suc_u,$ where $\suc_u$ is the preference associated with the node $u$. Suppose not, then there exist two candidates $x, y\in\CC$ such that $x\suc_u y, y\suc_i x,$ and $y\suc_j x$ for $i,j\in[3], i\ne j.$ But then the sub-profile of \PP along the path between $u_{\suc_i}$ and $u_{\suc_j}$ is not single-crossing with respect to this path. This contradicts our assumption that \PP is single-crossing with respect to the tree \TT.
\end{proof}

We now define the notion of {\em triad majority closure} of a profile which will be used crucially in our algorithm for \SCTR.

\begin{definition}[Triad Majority Closure]\label{def:triadclosure}
	The triad majority closure of any profile \PP is defined to be the inclusion-wise minimal profile $\overline{\PP}$ that satisfies the triad majority property and $\maj((\suc_i)_{i\in[3]})\in\overline{\PP}$ for every $\suc_1,\suc_2,\suc_3\in\PP$.
\end{definition}

The following result shows that a triad majority closure of a sub-profile of a profile that satisfies the triad majority property is unique thereby, establishing the well-definedness of \Cref{def:triadclosure}.

\begin{proposition}\label{prop:unique}
	If a profile \QQ with an odd number of profiles satisfies the triad majority property and $\maj((\suc_i)_{i\in[3]})\in\QQ$ for every $\suc_1,\suc_2,\suc_3\in\QQ$, then the triad majority closure of every sub-profile \PP of \QQ exists, and it is unique.
\end{proposition}

\begin{proof}
	Let $\MM(\PP)$ be the set of profiles which contain \PP and satisfy the triad majority property. We observe that $\QQ\in\MM(\PP)$. We define $\overline{\PP} = \cap_{\PP^\pr\in\MM(\PP)}\PP^\pr$ and claim that $\overline{\PP}$ is the triad majority closure of \PP. For two profiles $\PP_1, \PP_2$ that satisfy the triad majority property and contain \PP, the profile $\PP_1\cap\PP_2$ also satisfies the triad majority property and contains \PP. Hence, $\overline{\PP}$ satisfies the triad majority property and contains \PP; this proves existence. For uniqueness, let us assume, if possible, that two different profiles $\QQ_1$ and $\QQ_2$ are both triad majority closures of \PP. Then $\QQ_1\cap\QQ_2$ is also another triad majority closure of \PP. However, we have $\QQ_1\cap\QQ_2\subsetneq\QQ_1$ which contradicts our assumption that $\QQ_1$ is a triad majority closure of \PP. Hence, the triad majority closure of \PP exists, and it is unique.
\end{proof}

We now show that the triad majority closure of every weakly single-crossing profile on trees is unique.

\begin{lemma}\label{lem:unique}
	The triad majority closure of every weakly single-crossing profile on trees is unique and thus well-defined.
\end{lemma}

\begin{proof}
	Let \PP be any weakly single-crossing profile on trees. Hence, there exists a single crossing profile \QQ on trees such that \PP is a sub-profile of \QQ. From \Cref{lem:necessary}, we conclude that \QQ satisfies the triad majority property and $\maj((\suc_i)_{i\in[3]})\in\QQ$ for every $\suc_1,\suc_2,\suc_3\in\QQ$. Now using \Cref{prop:unique}, we conclude that the triad majority closure of \PP exists, unique, and thus well-defined.
\end{proof}

Similar to the triad majority closure, we next define the {\em single-crossing tree closure}.

\begin{definition}[Single Crossing Tree Closure]\label{def:scclosure}
	Let \PP be a weakly single-crossing profile on trees. The single-crossing tree closure of \PP is defined to be the inclusion-wise minimal profile $\PP^t$ containing \PP which is single-crossing with respect to some tree.
\end{definition}

The following result shows that, for every weakly single-crossing  profile \PP on trees, the single-crossing tree closure $\PP^t$ of \PP exists and is unique, and thus \Cref{def:scclosure} is well-defined. Moreover, $\PP^t$ turns out to be the triad majority closure of \PP.

\begin{theorem}\label{thm:triadclosure}
	Let \PP be a weakly single-crossing profile on trees. Then the triad majority closure $\overline{\PP}$ of \PP is single-crossing on trees. Therefore, the single-crossing tree closure of \PP exists, it is unique, and is equal to its triad majority closure $\overline{\PP}$. Moreover, the single-crossing tree closure of \PP can be computed in time polynomial in the size of $\PP$ and $|\overline{\PP}|\le|\PP|^3$.
\end{theorem}

\begin{proof}
	We can assume, without loss of generality, that the preferences in \PP are pairwise distinct thanks to \Cref{prop:distinct}. Let \RR be a profile of distinct preferences that is single-crossing with respect to a tree \TT such that \PP be a sub-profile of $\RR$ -- such an \RR exists since \PP is a weakly single-crossing profile on trees. Let us consider the subgraph $\TT^\pr = \cup_{\suc,\suc^\pr\in\PP} \ppp_{u_{\suc}, u_{\suc^\pr}}$ of \TT, where, for two nodes $u, v\in\TT$, $\ppp_{u,v}$ denotes the unique path between $u$ and $v$ in \TT. We observe that $\TT^\pr$ is actually a subtree of \TT since it is a connected subgraph of the tree \TT. Hence, the profile of preferences associated with the nodes in $\TT^\pr$ is single-crossing on trees. We also observe that, for every leaf node \lll of $\TT^\pr,$ the preference $\suc_\lll$ corresponding to \lll always belongs to \PP. We iteratively apply the following transformation on $\TT^\pr$ as long as we can: if there exists a node $u$ in $\TT^\pr$ of degree two such that the preference $\suc_u$ associated with $u$ does not belong to $\PP$, then we ``bypass'' $u$, that is, we remove $u$ from $\TT^\pr$ along with the two edges incident on $u$ and add an edge between the two neighbors of $u$. We denote the tree $\TT^\prr$ that results from $\TT^\pr$ after making all the transformations iteratively as long as we can. It follows directly from the definition of single-crossing profiles on trees that if the preference profile corresponding to $\TT^\pr$ is single-crossing on trees, then the preference profile corresponding to $\TT^\prr$ is single-crossing on trees.

	We claim that the triad majority closure $\overline{\PP}$ of $\PP$ is exactly the set \QQ of preferences associated with the nodes of the tree $\TT^\prr$. Clearly \QQ is also single-crossing on trees. We first observe that \QQ satisfies the triad majority property since \QQ is single-crossing with respect to the tree $\TT^\prr$ [see \Cref{lem:necessary}]. Now to show that \QQ is the triad majority closure of \PP, it is enough to show that every preference ${\suc}\in\QQ\setminus\PP$ is the majority order of some three preferences in \PP. Let ${\suc}\in\QQ\setminus\PP$ and $u\in\TT^\prr$ be the node in $\TT^\prr$ whose corresponding preference is \suc. We observe that the degree of $u$ in $\TT^\prr$ is at least $3,$ since otherwise the transformation would have deleted $u.$ We make $\TT^\prr$ rooted at $u$ and call it $\TT^\prr[u]$. Let $u_i, i\in[3],$ be any three neighbors of $u$ in $\TT^\prr.$ Let the subtrees of the rooted tree $\TT^\prr[u]$ rooted at $u_1, u_2,$ and $u_3$ be $\TT_1, \TT_2,$ and $\TT_3$ respectively. We observe that there must exist a node $v_i$ in $\TT_i$ such that the preference $\suc_{i}$ attached to $v_i$ belongs to \PP for every $i\in[3].$ We claim that $\suc$ is the majority order of $(\suc_i)_{i\in[3]}.$ Indeed, otherwise we may assume that there exist two candidates $x,y\in\CC$ such that $x\suc_u y$ and for at least two indices $i,j\in[3], i\ne j,$ we have $x\suc_i y$ and $x\suc_j y$. However, this violates the single-crossing property on the path between $v_2$ and $v_3$. Hence, $\suc$ is the majority order of $(\suc_i)_{i\in[3]}.$ This proves the claim. Hence, \QQ is the triad majority closure of \PP. Also \QQ is single-crossing with respect to the tree $\TT^\prr.$ Hence, by \Cref{lem:necessary}, \QQ is the single-crossing tree closure of \PP. We also observe that every preference in $\overline\PP\setminus\PP$ is a majority of some three preferences in \PP.

	%

	We consider the algorithm \AA that adds \nobreakseq{$\maj(\suc, \suc^\pr,\suc^\prr)$} iteratively to $\QQ$ (which is initialized to the empty set) for all three (not necessarily distinct) preferences $\suc,\suc^\pr,\suc^\prr\in\PP$, eliminates duplicates from \QQ, and outputs $\QQ$ so that every preference in \QQ are distinct. Clearly, we have $|\QQ|\le|\PP|^3$. The correctness of \AA follows from the fact that every preference in $\overline\PP\setminus\PP$ is a majority of some three (not necessarily distinct) preferences in \PP. The algorithm \AA runs in time polynomial in the size of $\overline{\PP}$ and thus polynomial in \PP since $|\overline\PP|\le|\PP|^3$.
\end{proof}

We will use \Cref{thm:triadclosure} in our \PE algorithm for single-crossing profiles on trees in \Cref{thm:scsequnknownub}. We believe that \Cref{thm:triadclosure} is an important structural result for weakly single-crossing profiles on trees which may be useful elsewhere also.

%


The following result on recognizing single-crossing profiles on trees is thanks to \cite{clearwater2015single}.

\begin{theorem}(Theorem 4.2 in~\cite{clearwater2015single})\label{thm:construction}
	Given a profile \PP, there is a polynomial time algorithm for checking whether there exists a tree \TT with respect to which \PP is single-crossing. Moreover, if a single-crossing tree exists, then the algorithm also outputs a tree $\TT^\pr$ with respect to which \PP is single-crossing.
\end{theorem}

\Cref{thm:construction,thm:triadclosure} give us a polynomial time algorithm for the \SCSR problem.

\begin{corollary}\label{cor:recognizing}
	The \SCTR problem is in \Pshort.
\end{corollary}

\begin{proof}
	Let \PP be any profile. We need to compute if \PP is weakly single-crossing. Using \Cref{prop:distinct}, let us assume that the preferences in \PP are pairwise distinct. Our algorithm \AA first tries to construct the triad majority closure $\overline{\PP}$ of \PP using the algorithm in \Cref{thm:triadclosure}. If \AA fails to construct the triad majority closure $\overline{\PP}$ of \PP, then we output \NO. While constructing $\overline{\PP}$ iteratively in \AA, if we have $|\overline{\PP}|>|\PP|^3$ at any point of time, we output \NO. Otherwise, \AA has constructed the triad majority closure $\overline{\PP}$ of \PP and $|\overline{\PP}|\le|\PP|^3$. Our algorithm now checks whether $\overline{\PP}$ is single-crossing with respect to some tree using \Cref{thm:construction}. If the algorithm in \Cref{thm:construction} outputs \YES, then \AA also outputs \YES and the tree returned by the algorithm in \Cref{thm:construction}; otherwise \AA outputs that the profile \PP is not weakly single-crossing on any tree.

	We now prove the correctness of our algorithm. Suppose the input profile \PP is indeed a weakly single crossing profile on trees. Then the triad majority closure $\overline{\PP}$ of \PP exists, $|\overline{\PP}|\le|\PP|^3$, and $\overline{\PP}$ is a single crossing profile on trees. It now follows from the correctness of the algorithm in \Cref{thm:construction} that our algorithm outputs \YES. Now let us assume that \PP is not a weakly single crossing profile on trees. If the triad majority closure $\overline{\PP}$ of \PP does not exist, then our algorithm correctly outputs \NO. If $\overline{\PP}$ exists but $|\overline{\PP}|>|\PP|^3$, then also our algorithm correctly outputs \NO. Now, since \PP is not a weakly single crossing profile on trees, there cannot exist another profile \QQ such that \PP is a sub-profile of \QQ and \QQ is single crossing on trees. Hence, $\overline{\PP}$ is not a single crossing profile on trees. It now follows from the correctness of  the algorithm in \Cref{thm:construction} that our algorithm outputs \NO. The polynomial running time of \AA follows from the polynomial running time of the algorithms in \Cref{thm:triadclosure,thm:construction}.
\end{proof}

\section{Elicitation Algorithm}

In this section, we present a polynomial time lower than $\OO(mn\log m)$ query complexity algorithm for eliciting weakly single-crossing profiles on trees. At a high level, our algorithm maintains a minimal single-crossing tree containing all the preferences that we have elicited so far. To elicit a new preference, we first search whether this new preference is present in the single-crossing tree that we have been maintaining with a small number of queries. If it is not present, then we elicit the new preference using $\OO(m\log m)$ queries by any efficient sorting algorithm, Heap sort for example. The crucial structural result that we will show is that any single-crossing preference profile cannot have too many different preferences. This helps us to bound the number of times we elicit a new preference using $\OO(m\log m)$ queries. We now present our algorithm in detail.

We begin with bounding the degree of any node in a single-crossing tree for a profile consisting of distinct preferences. We use this to bound the query complexity of our elicitation algorithm.

\begin{lemma}\label{lem:maxdegree}
	Let a profile $\PP$ of distinct preferences be single-crossing with respect to a tree $\TT$. Then the degree of every node in \TT is at most $m-1$.
\end{lemma}

\longversion{
	\begin{proof}
		Consider any node $v$ in \TT and let the preference $\suc_v$ associated with the node $v$ be $c_1\suc c_2\suc \cdots \suc c_m$. Let the neighbors of the node $v$ in \TT be $u_i, i\in[\el],$ and the preferences associated with them be $\suc_{u_i},i\in[\el],$ respectively. We need to show that $\el\le m-1$. We first observe that, for every $i\in[m-1],$ there exists at most one neighbor $u_j, j\in[\el],$ of the node $v$ such that the preference $\suc_{u_j}$ associated with the node $u_j$ has $c_{i+1} \suc_{u_j} c_i.$ Indeed, otherwise let us assume that there exist two neighbors $u_j$ and $u_{j^\pr}, j,j^\pr\in[\el], j\ne j^\pr,$ of the node $v$ such that the preferences $\suc_{u_j}$ and $\suc_{u_{j^\pr}}$ associated with them both have $c_{i+1} \suc_{u_j} c_i$ and $c_{i+1} \suc_{u_{j^\pr}} c_i.$ Then the sub-profile along the path $u_j, v, u_{j^\pr}$ is not single-crossing which contradicts our assumption that \PP is single-crossing with respect to the tree \TT. Hence, for every $i\in[m-1],$ there exists at most one neighbor $u_j, j\in[\el],$ of the node $v$ such that the preference $\suc_{u_j}$ associated with the node $u_j$ has $c_{i+1} \suc_{u_j} c_i.$ We also observe that, for every $j\in[\el],$ there exists an $i\in[m-1]$ such that $c_{i+1} \suc_{u_j} c_i$ since otherwise we have $\suc_v = \suc_{u_j}$ which contradicts our assumption that the preferences of the profile \PP are pairwise distinct. Hence, we have $\el\le m-1$ since otherwise we will have $i\in[m-1]$ and $j,j^\pr\in[\el], j\ne j^\pr,$ such that we have both $c_{i+1} \suc_{u_j} c_i$ and $c_{i+1} \suc_{u_{j^\pr}} c_i$ thanks to pigeon-hole argument which leads to a contradiction as argued above.
	\end{proof}
}

We remark that the assumption that the profile in \Cref{lem:maxdegree} consists of distinct preferences is crucial since a profile where all the preferences are the same is single-crossing with respect to every tree. We now show next a structural result which will be crucial for bounding query complexity of our elicitation algorithm.

\begin{lemma}\label{lem:balancednode}
	Let \TT be a tree of size at least $3$. Then there exists a node $r$ in \TT of degree at least $2$ such that the following holds: make the tree \TT rooted at $r$; let the neighbors of $r$ be $u_1, \ldots, u_\el$; let the subtrees of $\TT[r]$ rooted at $u_i, i\in[\el],$ be respectively $\TT_i, i\in[\el],$ and $|\TT_1|\ge |\TT_2|\ge \cdots \ge |\TT_\el|$; then $|\TT_1| \le 3|\TT|/4$.
\end{lemma}

\longversion{
	\begin{proof}
		Let $u_1$ be any arbitrary node in \TT. If $u_1$ satisfies the properties of the lemma, then we are done. Otherwise, let $u_2$ be the neighbor of $u_1$ such that, if we remove the edge $(u_1, u_2)$, then the connected component containing $u_2$ has more than $3|\TT|/4$ nodes. Let the removal of the edge $(u_1, u_2)$ creates two connected components $\VV_1$ and $\WW_1$ such that $u_1\in\VV_1$. If $u_2$ satisfies the properties of the lemma, then we are done. Otherwise, we repeat the above process defining $u_3$ and $\VV_2$. We observe that $\VV_1\subsetneq\VV_2$. We continue this process until we get a node that satisfies the properties of the lemma. The process has to terminate because otherwise we have an infinite chain $\VV_1\subsetneq\VV_2\subsetneq\VV_3\subsetneq\cdots\subseteq\TT$. This cannot happen since all the inclusions in the chain above are proper and $\VV_i\subseteq\TT$ for every $i$. Hence, the process always terminates with a node satisfying the properties of the lemma.
	\end{proof}
}

Given a profile \PP of distinct preferences, a tree \TT with respect to which \PP is single-crossing, and a preference ${\suc}\in\LL(\CC)$, we now present an algorithm for finding whether \suc belongs to \TT using $\OO(m)$ calls to \Query($\cdot$). This algorithm is a crucial component in our \PE algorithm.

\begin{algorithm}[!t]
	\caption{Searching preference in a single-crossing tree
		\label{alg:search}}
	\begin{algorithmic}[1]
		\Require{A profile $\PP = (\suc_i)_{i\in[n]}$ of distinct preferences, a tree \TT with respect to which \PP is single-crossing, and a preference ${\suc}\in\LL(\CC)$}
		\Ensure{\YES if \suc belongs to \PP and \NO otherwise}
		\While{$|\TT|\ge 3$}\label{alg:searchwhileloop}
		\State{$r \leftarrow$ a node in \TT as in \Cref{lem:balancednode}. $u_i, i\in[\el], \el\ge 2,$ be the neighbors of $r$. Let $\TT_i, i\in[\el],$ be the subtrees of $\TT[r]$ rooted at $u_i, i\in[\el],$ respectively such that $|\TT_1|\ge |\TT_2|\ge \cdots \ge |\TT_\el|$.}
		\For{$i=1$ to \el}\label{alg:searchforloop}
		\State{Let $x_i, y_i\in\CC$ such that $x_i \suc_r y_i$ and $y_i \suc_{u_i} x_i$}
		\If{\Query($x_i \suc y_i$)=\true}
		\State{$\TT\leftarrow\TT\setminus\TT_i$}\label{alg:searchremovesubtree}
		\Else
		\State{$\TT\leftarrow\TT_i$ and exit for loop}\label{alg:searchsubtree}
		\EndIf
		\EndFor
		\EndWhile\label{alg:searchwhileloopend}
		\If{\suc belongs to \TT}\label{alg:ifbegin}
		\Return \YES\Comment{Can be done in $\OO(m)$ queries.}
		\Else \;\Return \NO
		\EndIf\label{alg:ifend}
	\end{algorithmic}\label{alg1}
\end{algorithm}

\begin{lemma}\label{lem:search}
	Let \TT be a tree on which an unknown profile \PP of distinct preferences is single-crossing. Given an oracle access to a preference ${\suc}\in\LL(\CC)$ using \Query($\cdot$), \Cref{alg1} checks if \suc belongs to \PP in polynomial-time using $\OO(m)$ calls to \Query($\cdot$).
\end{lemma}

\begin{proof}
	We present our algorithm in \Cref{alg:search}. Let $r$ be a node as in \Cref{lem:balancednode}. We consider the tree \TT rooted at $r$ which we denote by $\TT[r]$. Let the neighbors of $r$ in $\TT[r]$ be $u_1, \ldots, u_\el$, for some $2\le\el\le m-1$ (the upper bound of $m-1$ follows from \Cref{lem:maxdegree}), the rooted subtrees of $\TT[r]$ rooted at $u_i, i\in[\el],$ be $\TT_i, i\in[\el], |\TT_1|\ge |\TT_2|\ge \cdots \ge |\TT_\el|$, and $|\TT_1| \le 3|\TT|/4$. Since the profile \PP consists of distinct preferences, for every $i\in[\el]$, there exist two candidates $x_i,y_i\in\CC$ such that $x_i\suc_r y_i$ and $y_i\suc_{u_i} x_i$. Let \suc be the preference which we have to search in \TT; we can access \suc only through the \Query($\cdot$) function. We query the oracle for $x_i$ versus $y_i$ in $\suc.$ If $x_i\suc y_i$, then \suc cannot belong to the subtree $\TT_i$, and thus we remove $\TT_i$ from \TT as done in line \ref{alg:searchremovesubtree} of \Cref{alg:search}. On the other hand, if $y_i\suc x_i$, then \suc cannot belong to $\TT\setminus\TT_i$, and thus we remove $\TT\setminus\TT_i$ from \TT as is done in line \ref{alg:searchsubtree} of \Cref{alg:search}. Hence, after each iteration of the while loop, \suc remains in \TT if \suc was present at the start of the iteration. Since every iteration decreases the size of \TT, the algorithm terminates and is correct. We note that a similar argument was used in the proof of \Cref{thm:triadclosure}.

	We now turn our attention to the query complexity of \Cref{alg:search}. For a tree \TT with $n$ nodes where every preference is over $m$ candidates, let $T(n,m)$ be the query complexity of the while loop from line \ref{alg:searchwhileloop} to line \ref{alg:searchwhileloopend} in \Cref{alg:search}. Let us consider an iteration of the while loop at line \ref{alg:searchwhileloop}. Let $k$ be the number of times the for loop at line \ref{alg:searchforloop} iterates. If $k=1$, then the algorithm makes only one query in the current iteration of the while loop in \Cref{alg:search} and the number of nodes in the new tree \TT is upper bounded by $3/4$-th time the number of nodes in the previous \TT (by the choice of $r$). For $k\ge 2$, we observe that the number of nodes in the new tree \TT is upper bounded by $1/k$-th times the number of nodes in the previous \TT -- this is because the for loop at line \ref{alg:searchforloop} iterates according to a nonincreasing order of subtrees. Hence, we have the following recurrence relation.
	\[T(n,m) \le \max\left\{T(3n/4,m)+1, \max_{k=\{2, \ldots, m\}}\left\{T(n/k,m)+k\right\}\right\},\]\[ T(2,m) = T(1,m) = \OO(m)\]
	We know from \Cref{lem:sizeupperbound} that $n\le {m\choose 2} + 1$. By solving the above recurrence with the ${m\choose 2} + 1$ upper bound on $n$ and the fact that line \ref{alg:ifbegin} to \ref{alg:ifend} can be executed with $\OO(m)$ queries (see \cite[Lemma 2]{deycross}), we get that the query complexity of \Cref{alg:search} is $\OO(m)$.
\end{proof}

\begin{algorithm}[!t]
	\caption{Eliciting a profile which is single-crossing with respect to some unknown tree
		\label{alg:elicitation}}
	\begin{algorithmic}[1]
		\Require{$\pi$ be the order in which voters arrive}
		\Ensure{Profile of all the voters}
		\State $\RR, \QQ \leftarrow \emptyset$ \LineComment{\QQ stores all the votes seen so far without duplicate. \RR stores the profile.}
		\For{$i \gets 1 \textrm{ to } n$} \LineComment{Elicit preference of the $i^{th}$ voter in $i^{th}$ iteration of this for loop.}
		\State{\TT be the tree with respect to which the single-crossing tree closure $\overline{\QQ}$ of \QQ is single-crossing}
		\If{$\succ_{\pi(i)}=\suc$ for some ${\suc}\in\overline{\QQ}$}\LineComment{Can be done using $\OO(m)$ queries by \Cref{lem:search}}
		\State $\RR[\pi(i)] \leftarrow \suc$
		\Else
		\State $\RR[\pi(i)] \leftarrow$ Elicit using \Cref{obs:naive}\label{alg:usingnaive}
		\State $\QQ \leftarrow \QQ\cup\{\RR[\pi(i)]\}$
		\EndIf
		\EndFor
	\end{algorithmic}
\end{algorithm}

We first present a \PE algorithm for single-crossing profiles on trees when we are given a sequential access to the preference (the sequential order is a priori not known) and we do not know any tree with respect to which the input profile is single-crossing.

\begin{theorem}\label{thm:scsequnknownub}
	Suppose a profile \PP is single-crossing with respect to some unknown tree \TT. Suppose we have only sequential access to the preferences whose ordering is also not known a priori. Then there is a \PE algorithm for the single-crossing profiles on trees with query complexity $\BigO(mn + \min\{m^2, n\} m\log m)$.
\end{theorem}

\begin{proof}
	We present our \PE algorithm in \Cref{alg:elicitation}. We maintain an array \RR of length $n$ to store the preferences of all the $n$ voters and a set \QQ to store all the distinct preferences seen so far. Let $\pi$ be the order in which the voters are accessed. To elicit the preference of voter $\pi(i)$, we first construct, in polynomial time, a tree \TT with respect to which the single-crossing tree closure $\overline{\QQ}$ of \QQ is single-crossing using \Cref{thm:triadclosure}. Next we find whether the preference of the voter $\pi(i)$ is already present in \TT; this can be done in polynomial time with making $\OO(m)$ calls to $\text{\Query}(\cdot)$ using \Cref{lem:search}. If the preference $\suc_{\pi(i)}$ of the voter $\pi(i)$ is present in \TT, we have elicited $\suc_{\pi(i)}$ using $\OO(m)$ queries. Otherwise, we elicit $\suc_{\pi(i)}$ using $\OO(m\log m)$ queries by \Cref{obs:naive}. However, since the number of distinct preferences in any profile which is single-crossing with respect to some tree is at most ${m\choose 2}+1$ thanks to \Cref{lem:sizeupperbound}, we use \Cref{obs:naive} at most ${m\choose 2}+1$ times. Hence, the query complexity of \Cref{alg:elicitation} is $\BigO(mn + \min\{m^2, n\} m\log m)$.
\end{proof}

Our algorithm in \Cref{thm:scsequnknownub} can readily be seen to work for weakly single-crossing profiles on trees too. Hence, we have the following corollary.

\begin{corollary}\label{cor:elicitsubdomain}
	There exists a polynomial time \PE algorithm for the weakly single-crossing profiles on trees with query complexity $\BigO(mn + \min\{m^2, n\} m\log m)$.
\end{corollary}

We prove next that the query complexity upper bound in \Cref{thm:scsequnknownub} is optimal up to constant factors for a large number of voters (when $n = \Omega(m^2\log m)$), even if a tree \TT with respect to which the input profile \PP is single-crossing, is known and random access to preferences are allowed.

\begin{theorem}\label{thm:scrandomknownlb}
	Let a profile \PP be single-crossing with respect to a tree \TT. Let \TT be known except for the preferences associated with the nodes of \TT. Then any \PE algorithm has query complexity $\Omega(\max\{m\log m, mn\})$ even if we are allowed to query preferences randomly and interleave the queries to different preferences arbitrarily.
\end{theorem}

\begin{proof}
	The $\Omega(m\log m)$ bound follows from the sorting lower bound and the fact that any profile consisting of only one preference $P\in\LL(\CC)$ is single-crossing. Let \TT be a star with one central node and $n-1$ leaf nodes attached to the central node with an edge. Suppose we have an even number of candidates, that is, $\CC = \{c_1, \ldots, c_m\}$ for some even integer $m$. Consider the ordering $\QQ = c_1 \succ c_2 \succ \cdots \succ c_m$ and the pairing of the candidates $\{c_1, c_2\}, \{c_3, c_4\}, \ldots, \{c_{m-1}, c_m\}$. The oracle answers all the query requests consistently according to the ordering $\QQ$. We claim that any \PE algorithm \AA must compare $c_i$ and $c_{i+1}$ for every voter corresponding to every leaf node of \TT and for every odd integer $i\in[m]$. Indeed, otherwise, there exists a voter $\kappa$ corresponding to a leaf node of \TT and an odd integer $i\in[m]$ such that the algorithm \AA does not compare $c_i$ and $c_{i+1}$. Suppose the algorithm outputs a profile $\PP^\pr$. The oracle fixes the preference of every voter except $\kappa$ to be \QQ. If the voter $\kappa$ in $\PP^\pr$ prefers $c_i$ over $c_{i+1}$ in $\PP^\pr$, then the oracle fixes the preference $\suc_\kappa$ to be $c_1 \succ c_2 \succ \cdots \succ c_{i-1} \succ c_{i+1} \succ c_i \succ c_{i+2} \succ \cdots \succ c_m$; otherwise the oracle fixes $\suc_\kappa$ to be \QQ. The algorithm fails to correctly output the preference of the voter $\kappa$ in both the cases. Also the final profile with the oracle is single-crossing with respect to the tree \TT. Hence, \AA must compare $c_i$ and $c_{i+1}$ for every voter corresponding to every leaf node of \TT and for every odd integer $i\in[m]$ and thus has query complexity $\Omega(mn)$.
\end{proof}

Hence, the query complexity of \PE for the weakly single-crossing profiles on trees, does not depend substantially on how the preferences are accessed and whether we know a single-crossing tree of the single-crossing closure of the input profile. This is in sharp contrast to the corresponding results for the single-crossing domain where the query complexity for \PE improves substantially if we know a single-crossing ordering and we have a random access to the preferences than if we only have sequential access to preferences~\cite{deycross}. Next, we present our lower bound on the query complexity of any \PE algorithm for single-crossing domains in the random access model.

\begin{theorem}\label{thm:lb_sc}
	Let a profile \PP be single-crossing with respect to the identity permutation of the voters. Then any algorithm for eliciting \PP has query complexity $\Omega(m^2\log n)$ even if we are allowed to query preferences randomly.
\end{theorem}

\begin{proof}
	Any \PE algorithm can be viewed to be a binary decision tree. Let the set \CC of candidates be $\{c_i: i\in[m]\}$. Let \QQ be any single-crossing profile consisting of $\el=m(m-1)/2$ distinct preferences; we know such a profile exists. Let \RR be the collection of all profiles of size $n$ which contains multiple copies of every preference present in \QQ and nothing else. Then we have $|\RR|=\el^n$. Since every profile in \RR can be a valid input to any \PE algorithm \AA, the corresponding binary decision tree must have at least $\el^n$ leaf vertices. Hence, the query complexity of \AA is the height of the corresponding binary decision tree which is at least $\Omega(m^2\log n)$.
\end{proof}

\section{Conclusion and Future Work}

We have shown that weakly single-crossing profiles on trees inherit many desirable properties from single-crossing profiles on trees like transitivity of the majority relation [\Cref{prop:tran}]. This implies that we can compute the winner under many intractable voting rules like Kemeny, Dodgson, etc., in polynomial time for weakly single-crossing profiles. We have also presented an efficient algorithm to recognize weakly single-crossing profiles. We then showed how weakly single-crossing profiles on trees can be elicited with a small number of queries to the voters. Last but not the least, we have resolved an open question in~\cite{deycross} and thereby proving optimality of their preference elicitation algorithm when random queries are allowed.

An immediate future work is to extend our recognition and elicitation algorithms to the more general median graphs. Characterizing weakly single-crossing profiles on trees in terms of forbidden structures is another important direction of research.

\section*{Acknowledgements}

Palash Dey thanks the Science, Education, and Research Board (SERB), Government of India, for supporting this work through Core Research Grant under file no. CRG/2022/003294. He also thanks the invaluable feedback received from the reviewers of Theoretical Computer Science.

\bibliographystyle{alpha}
\bibliography{references}

\newcommand{\etalchar}[1]{$^{#1}$}
\begin{thebibliography}{MCWG95}

\bibitem[BBHH15]{brandt2015bypassing}
Felix Brandt, Markus Brill, Edith Hemaspaandra, and Lane~A Hemaspaandra.
\newblock Bypassing combinatorial protections: Polynomial-time algorithms for
  single-peaked electorates.
\newblock {\em J. Artif. Intell. Res.}, pages 439--496, 2015.

\bibitem[BC84]{blair1984labor}
Douglas~H Blair and David~L Crawford.
\newblock Labor union objectives and collective bargaining.
\newblock {\em Q. J. Econ.}, pages 547--566, 1984.

\bibitem[BCW16]{DBLP:journals/mss/BredereckCW16}
Robert Bredereck, Jiehua Chen, and Gerhard~J. Woeginger.
\newblock Are there any nicely structured preference profiles nearby?
\newblock {\em Math. Soc. Sci.}, 79:61--73, 2016.

\bibitem[Bla48]{black1948rationale}
Duncan Black.
\newblock On the rationale of group decision-making.
\newblock {\em J. Polit. Econ.}, pages 23--34, 1948.

\bibitem[BM11]{barbera2011top}
Salvador Barber{\`a} and Bernardo Moreno.
\newblock Top monotonicity: A common root for single peakedness, single
  crossing and the median voter result.
\newblock {\em GEB}, 73(2):345--359, 2011.

\bibitem[BNM{\etalchar{+}}58]{black1958theory}
Duncan Black, Robert~Albert Newing, Iain McLean, Alistair McMillan, and Burt~L
  Monroe.
\newblock {\em The theory of committees and elections}.
\newblock Springer, 1958.

\bibitem[BSU13]{betzler2013computation}
Nadja Betzler, Arkadii Slinko, and Johannes Uhlmann.
\newblock On the computation of fully proportional representation.
\newblock {\em J. Artif. Intell. Res.}, 47:475--519, 2013.

\bibitem[CGS12]{CornazGS12}
Denis Cornaz, Lucie Galand, and Olivier Spanjaard.
\newblock Bounded single-peaked width and proportional representation.
\newblock In {\em Proc. Twentieth European Conference on Artificial
  Intelligence ({ECAI})}, volume 242 of {\em Frontiers in Artificial
  Intelligence and Applications}, pages 270--275. {IOS} Press, 2012.

\bibitem[Con09]{Conitzer09}
Vincent Conitzer.
\newblock Eliciting single-peaked preferences using comparison queries.
\newblock {\em J. Artif. Intell. Res.}, 35:161--191, 2009.

\bibitem[Cor09]{cormen2009introduction}
Thomas~H Cormen.
\newblock {\em Introduction to algorithms}.
\newblock MIT press, 2009.

\bibitem[CPS15]{clearwater2015single}
Adam Clearwater, Clemens Puppe, and Arkadii Slinko.
\newblock Generalizing the single-crossing property on lines and trees to
  intermediate preferences on median graphs.
\newblock In {\em Proc. 24th International Conference on Artificial
  Intelligence}, IJCAI, pages 32--38. AAAI Press, 2015.

\bibitem[CSS99]{Cohen}
William~W. Cohen, Robert~E. Schapire, and Yoram Singer.
\newblock Learning to order things.
\newblock {\em J. Artif. Int. Res.}, 10(1):243--270, May 1999.

\bibitem[Dem82]{demange1982single}
Gabrielle Demange.
\newblock Single-peaked orders on a tree.
\newblock {\em Mathematical Social Sciences}, 3(4):389--396, 1982.

\bibitem[Dem12]{demange2012majority}
Gabrielle Demange.
\newblock Majority relation and median representative ordering.
\newblock {\em SERIEs}, 3(1-2):95--109, 2012.

\bibitem[Die97]{diestelBook}
Reinhard Diestel.
\newblock {\em Graph theory}.
\newblock Graduate texts in mathematics. Springer, New York, Berlin, Paris,
  1997.

\bibitem[DKNS01]{Dwork}
Cynthia Dwork, Ravi Kumar, Moni Naor, and D.~Sivakumar.
\newblock Rank aggregation methods for the web.
\newblock In {\em Proc. 10th International Conference on World Wide Web}, WWW
  '01, pages 613--622, New York, NY, USA, 2001. ACM.

\bibitem[DM16a]{deypeak}
Palash Dey and Neeldhara Misra.
\newblock Elicitation for preferences single peaked on trees.
\newblock In {\em Proc. Twenty-Fifth International Joint Conference on
  Artificial Intelligence, {IJCAI} 2016, New York, NY, USA, 9-15 July 2016},
  pages 215--221, 2016.

\bibitem[DM16b]{deycross}
Palash Dey and Neeldhara Misra.
\newblock Preference elicitation for single crossing domain.
\newblock In {\em Proc. Twenty-Fifth International Joint Conference on
  Artificial Intelligence, {IJCAI} 2016, New York, NY, USA, 9-15 July 2016},
  pages 222--228, 2016.

\bibitem[Dod76]{dodgson1876method}
Charles~Lutwidge Dodgson.
\newblock {\em A method of taking votes on more than two issues}.
\newblock 1876.

\bibitem[DS74]{diamond1974increases}
Peter~A Diamond and Joseph~E Stiglitz.
\newblock Increases in risk and in risk aversion.
\newblock {\em J. Econ. Theory}, 8(3):337--360, 1974.

\bibitem[EFLO15]{DBLP:conf/aaai/ElkindFLO15}
Edith Elkind, Piotr Faliszewski, Martin Lackner, and Svetlana Obraztsova.
\newblock The complexity of recognizing incomplete single-crossing preferences.
\newblock In {\em Proc. Twenty-Ninth {AAAI} Conference on Artificial
  Intelligence, January 25-30, 2015, Austin, Texas, {USA.}}, pages 865--871,
  2015.

\bibitem[EFS20]{DBLP:journals/scw/ElkindFS20}
Edith Elkind, Piotr Faliszewski, and Piotr Skowron.
\newblock A characterization of the single-peaked single-crossing domain.
\newblock {\em Soc. Choice Welf.}, 54(1):167--181, 2020.

\bibitem[EL14]{DBLP:conf/aaai/ElkindL14}
Edith Elkind and Martin Lackner.
\newblock On detecting nearly structured preference profiles.
\newblock In {\em Proc. Twenty-Eighth {AAAI} Conference on Artificial
  Intelligence, July 27 -31, 2014, Qu{\'{e}}bec City, Qu{\'{e}}bec, Canada.},
  pages 661--667, 2014.

\bibitem[ELP17]{DBLP:conf/aaai/ErdelyiLP13}
G{\'{a}}bor Erd{\'{e}}lyi, Martin Lackner, and Andreas Pfandler.
\newblock Computational aspects of nearly single-peaked electorates.
\newblock {\em J. Artif. Intell. Res.}, 58:297--337, 2017.

\bibitem[EMP22]{DBLP:journals/corr/abs-2205-09092}
Edith Elkind, Martin, and Dominik Peters.
\newblock Preference restrictions in computational social choice: {A} survey.
\newblock {\em CoRR}, abs/2205.09092, 2022.

\bibitem[FHH14]{Faliszewski:2014:CMA:2566265.2566291}
Piotr Faliszewski, Edith Hemaspaandra, and Lane~A. Hemaspaandra.
\newblock The complexity of manipulative attacks in nearly single-peaked
  electorates.
\newblock {\em Artif. Intell.}, 207:69--99, February 2014.

\bibitem[FHHR11]{DBLP:journals/iandc/FaliszewskiHHR11}
Piotr Faliszewski, Edith Hemaspaandra, Lane~A. Hemaspaandra, and J{\"{o}}rg
  Rothe.
\newblock The shield that never was: Societies with single-peaked preferences
  are more open to manipulation and control.
\newblock {\em Inf. Comput.}, 209(2):89--107, 2011.

\bibitem[FL20]{Lackner14}
Zack Fitzsimmons and Martin Lackner.
\newblock Incomplete preferences in single-peaked electorates.
\newblock {\em J. Artif. Intell. Res.}, 67:797--833, 2020.

\bibitem[HM97]{hinich1997analytical}
Mel vin~J Hinich and Michael~C Munger.
\newblock {\em Analytical politics}.
\newblock Cambridge University Press, 1997.

\bibitem[JSA08]{JacksonSA08}
Benjamin~G. Jackson, Patrick~S. Schnable, and Srinivas Aluru.
\newblock Consensus genetic maps as median orders from inconsistent sources.
\newblock {\em {IEEE/ACM} Trans. Comput. Biology Bioinform.}, 5(2):161--171,
  2008.

\bibitem[Kem59]{kemeny1959mathematics}
John~G Kemeny.
\newblock Mathematics without numbers.
\newblock {\em Daedalus}, 88(4):577--591, 1959.

\bibitem[Kun15]{kung2015sorting}
Fan-Chin Kung.
\newblock Sorting out single-crossing preferences on networks.
\newblock {\em Soc. Choice Welf.}, 44(3):663--672, 2015.

\bibitem[Lev75]{levenglick1975fair}
Arthur Levenglick.
\newblock Fair and reasonable election systems.
\newblock {\em Behavioral Science}, 20(1):34--46, 1975.

\bibitem[MCWG95]{mas1995microeconomic}
A.~Mas-Colell, M.D. Whinston, and J.R. Green.
\newblock {\em Microeconomic Theory}.
\newblock Microeconomic Theory. Oxford University Press, 1995.

\bibitem[MG04]{merlin2004domain}
Vincent Merlin and Wulf Gaertner.
\newblock Domain conditions in social choice theory, 2004.

\bibitem[Mir71]{mirrlees1971exploration}
James~A Mirrlees.
\newblock An exploration in the theory of optimum income taxation.
\newblock {\em Rev. Econ. Stud.}, pages 175--208, 1971.

\bibitem[Mou91]{moulin1991axioms}
Hervi Moulin.
\newblock {\em Axioms of cooperative decision making}.
\newblock Number~15. Cambridge University Press, 1991.

\bibitem[MR81]{meltzer1981rational}
Allan~H Meltzer and Scott~F Richard.
\newblock A rational theory of the size of government.
\newblock {\em J. Polit. Econ.}, pages 914--927, 1981.

\bibitem[Rob77]{roberts1977voting}
Kevin~WS Roberts.
\newblock Voting over income tax schedules.
\newblock {\em J Public Econ}, 8(3):329--340, 1977.

\bibitem[SYFE13]{skowron2013complexity}
Piotr Skowron, Lan Yu, Piotr Faliszewski, and Edith Elkind.
\newblock The complexity of fully proportional representation for
  single-crossing electorates.
\newblock In {\em International Symposium on Algorithmic Game Theory}, pages
  1--12. Springer, 2013.

\end{thebibliography}

\end{document}